\documentclass[conference]{IEEEtran}
 \usepackage{amsmath}
   \usepackage{amsfonts}   
   \usepackage{amssymb}    
\usepackage[amsmath,thmmarks]{ntheorem}
\usepackage{subfigure}

\makeatletter
\def\ps@headings{%
\def\@oddhead{\mbox{}\scriptsize\rightmark \hfil \thepage}%
\def\@evenhead{\scriptsize\thepage \hfil \leftmark\mbox{}}%
\def\@oddfoot{}%
\def\@evenfoot{}}
\makeatother
\ifCLASSINFOpdf
  \usepackage[pdftex]{graphicx}
\else
  \usepackage[dvips]{graphicx}
\fi

\begin{document}

\newtheorem{thm}{Theorem}[section]
\newtheorem{cor}[thm]{Corollary}
\newtheorem{lem}[thm]{Lemma}
\newtheorem{prop}[thm]{Proposition}

\newtheorem{rem}[thm]{Remark}

\newtheorem{ex}[thm]{Example}

\newtheorem{defi}[thm]{Definition}

%
\title{Comparative Resilience Notions and Vertex Attack Tolerance of Scale-Free Networks}


\author{
  \IEEEauthorblockN{John Matta, Jeffrey Borwey, Gunes Ercal}
  \IEEEauthorblockA{Southern Illinois University Edwardsville
  \\\{jmatta,jborwey,gercal\}@siue.edu}
}

\maketitle

\begin{abstract}

We are concerned with an appropriate mathematical measure of resilience in the face of targeted node attacks for arbitrary degree networks, and subsequently comparing the resilience of different scale-free network models with the proposed measure.  We strongly motivate our resilience measure termed \emph{vertex attack tolerance} (VAT), which is denoted mathematically as $\tau(G) =  \min_{S \subset V} \frac{|S|}{|V-S-C_{max}(V-S)|+1}$, where  $C_{max}(V-S)$ is the largest connected component in $V-S$.  We attempt a thorough comparison of VAT with several existing resilience notions: conductance, vertex expansion, integrity, toughness, tenacity and scattering number.  Our comparisons indicate that for artbitrary degree distributions VAT is the only measure that fully captures both the major \emph{bottlenecks} of a network and the resulting \emph{component size distribution} upon targeted node attacks (both captured in a manner proportional to the size of the attack set).  For the case of $d$-regular graphs, we prove that $\tau(G) \le d\Phi(G)$, where $\Phi(G)$ is the conductance of the graph $G$.  Conductance and expansion are well-studied measures of robustness and bottlenecks in the case of regular graphs but fail to capture resilience in the case of highly heterogeneous degree graphs.  Regarding comparison of different scale-free graph models, our experimental results indicate that PLOD graphs with degree distributions identical to BA graphs of the same size exhibit consistently better vertex attack tolerance than the BA type graphs, although both graph types appear asymptotically resilient for BA generative parameter $m = 2$.  BA graphs with $m = 1$ also appear to lack resilience, not only exhibiting very low VAT values, but also great transparency in the identification of the vulnerable node sets, namely the hubs, consistent with well known previous work.

\end{abstract}


%
\IEEEpeerreviewmaketitle

\section{Introduction and Related Work}

While networks arise ubiquitously from countless and varied disciplines, an important problem relevant to all types of networks is that of measuring and generating resilience characteristics.  Although the definition of what constitutes a resilient network may differ across domains (particularly if the network's nodes and edges have various attributes associated with them), the topological properties characterizing resilience may be fairly stable across divergent domains and may depend only on the network's graph theoretic representation.  At a minimum, any notion of resilience must express the relative size of a most critical set of target edges or target vertices whose removal (upon an attack) would be detrimental to the remaining network and attempt to quantify the amount of damage that is done.  The graph-theoretic network resilience problem in its utmost generality is not new.  But, the most studied theoretical machinery  to discuss it, namely conductance\cite{ChungSpectraBook} (and the related expansion \cite{Alon86}), is more meaningful for networks with homogeneous degree distributions because it is a fundamentally edge-based notion.  However, there is a growing body of work \cite{Newman2006} indicating that many real world networks, ranging from biological networks to online social networks, often exhibit heterogeneous and scale-free degree distributions.  Even in situations where the scale-free nature of a degree distribution may be questioned \cite{DoylePNAS}, an assumption of degree homogeneity would be even more suspect in many networks \cite{aldersonTopo}.  Some notable exceptions to this are, of course, networks whose connection are purely defined by distances in low dimensional space (such as random ad-hoc networks) as well as networks that have very hard and frugal constraints on the number of connections their nodes are allowed.  In any case, there is a gap between much of the beautiful existing graph theoretical machinery to analyze resilience in networks based on conductance and the actual resilience, particularly against \emph{node} attacks in many actual networks, which have heterogeneous degrees.

It is important to note, of course, that conductance (in its various normalized and non-normalized forms) is a fundamental property of graphs independent of its applicability to edge-based resilience or resilience in general.  Conductance is intimately related to both the mixing times of random walks on graphs \cite{SinclairJerrum} and the eigenvalues of the graph's adjacency matrix (or normalized adjacency matrix as befits the situation) \cite{ChungSpectraBook}.  Therefore, one very important application of conductance has involved the analysis of randomized algorithms, particularly for hard counting problems, yielding an important place for the measure in theoretical computer science.  Many of the Markov chains considered in such contexts are regular or almost-regular, and it is relatively recently that the heterogeneity of many actual networks has become established across many fields, the relevance of such degree distributions first being established not by theoretical computer scientists but by \emph{network scientists} in any case.

To illustrate why degree heterogeneity poses such a problem in the application of conductance (conveniently approximable via eigenvalues) to resilience, we first remind the reader of the famous work \cite{ABAchilles} claiming the vulnerability of the internet to targeted attacks and the equally famous response \cite{DoylePNAS} and extension \cite{aldersonTopo} questioning the assumptions of the model considered.  An observer may then conclude that although the Barabasi-Albert model (and preferential attachment in general) may tend to induce a small target set of critical nodes whose removal severely disconnects the graph, there may be other generative models of scale-free networks that are significantly more resilient (against node attacks).  Particularly, an observer familiar with the conductance benefits of \emph{completely random} edge choices for homogeneous degree networks \cite{Furedi81eigenvalues,Friedman89,Alon86} may even be tempted to claim that a generative model in which edges are randomly chosen conditional upon the constraint of a pre-specified power-law degree distribution (such as the PLOD model \cite{Palmer2000} ) should be more resilient than the preferential attachment model.  Yet, if one were to take conductance as the measure of resilience against node attacks, then one would observe a precisely contradicting response: In fact, the conductance of scale-free graphs generated via preferential attachment (which we shall refer to as BA graphs henceforth) is \emph{optimal} \cite{conductancePASF} while the conductance of random scale-free graphs is sub-optimal \cite{conductanceRandSF} (especially relative to BA graphs) though still fairly good.  Clearly, conductance does not capture resilience against node attacks.

In this work, we are motivated by the search for both (i) a measure that \emph{does} capture resilience against node attacks in a meaningful way, as well (ii) the actual comparative resilience of different generative scale-free models.  As such, this work is composed of two logical modules. In the first module, we propose and strongly motivate our resilience measure which we call \emph{vertex attack tolerance} denoted by mathematically by $\tau(G)$ and shortened VAT.  Our measure was first introduced by the authors in preliminary work \cite{CASResSF}, wherein some examples were used to compare and contrast the measure with conductance.  We note that in this work, we slightly modify the measure to smooth it by adding one to the denominator.  But, more importantly, we attempt a thorough comparison of vertex attack tolerance with several existing node-based resilience notions found via a thorough graph theory literature survey.  In particular, we compare and contrast VAT with the following previously known graph-theoretic notions of resilience: edge-based conductance\cite{SinclairJerrum,ChungSpectraBook}, vertex expansion\cite{vertExpApprox,vertexSeparators}, integrity\cite{Barefoot1987}, toughness\cite{Chvatal2006}, tenacity\cite{Cozzens1995} and scattering number\cite{Jung1978}.  While some such measures appear surprisingly similar in nominal form to vertex attack tolerance, we show that VAT captures resilience in a strictly more meaningful way.  Furthermore, we mathematically relate $\tau(G)$ to edge conductance by appropriate lower and upper bounds in the case of regular graphs, which is a case in which conductance is known to be meaningful.  One advantage of such relationships in the case of regular graphs is that eigenvalue bounds follow as corollaries as well.  We also present other bounds on VAT with respect to other measures in special cases.  However, we demonstrate and maintain that vertex attack tolerance is a strictly more appropriate resilience measure for both homogeneous and heterogeneous degree graphs in the context of node failures.

The second logical module of our paper then consists of extensive calculations of $\tau(G)$ for important generative models of scale-free graphs: preferential attachment graphs captured via the BA model, random scale-free graphs captured via the PLOD model, and \emph{heuristically-optimized trade-offs} graphs referred to as the HOTNet model\cite{HOTNet}.  Due to the inherent intractability (i.e. NP-completeness and in some cases approximation-hardness) of many aforementioned resilience notions, as well as the plausibly more difficult combinatorial form that vertex attack tolerance takes, an efficient exact algorithm to compute it is unlikely.  Therefore, we have only been able to exactly compute $\tau(G)$ of networks considered for sizes up to 40 nodes via a branch-and-bound approach.  Beyond that size and up to 2500 nodes, we used genetic algorithms and local search methods to calculate as accurately as feasible the vertex attack tolerances.  The methods themselves, and how they were optimized to converge more quickly, are of algorithmic interest in their own right.  Of course, the exact settings of the parameters greatly alters the graph structure, and one advantage of the PLOD model in that respect is that the degree distribution may be generated by any other model a priori.  Unfortunately, HOTNets, which exhibited the worst vertex attack tolerance on the various parameter settings considered even when controlling for average degree, cannot be tweaked to follow an exact degree distribution as such.  Furthermore, for small graph sizes, BA graphs exhibit better vertex attack tolerance than both PLOD and HOTNet models.  However it is the asymptotic behavior of a graph family that is statistically relevant.  Our results indicate that PLOD graphs with degree distributions identical to (i.e. with degree distributions generated from) BA graphs of the same size exhibit increasingly better vertex attack tolerance than the BA type graphs, although both graph types appear surprisingly resilient when the generatize BA parameter is $m = 2$.  In particular, the VAT values diminish very slowly despite repeated doubling of graph sizes.  We also discovered the following in the process of our simulations: \emph{Finding} bad attack sets (sets of nodes whose removal causes great disruption according to the VAT measure) were extremely difficult to find for BA graphs.  This leads to our current questions regarding resilience with respect to specific attack and failure models.

\section{Definitions and Preliminaries}

We define vertex attack tolerance (VAT) as
\begin{equation*}
\tau(G) = \min_{S \subset V} \{\frac{|S|}{|V-S-C_{max}(V-S)|+1 } \} 
\end{equation*}
where $C_{max}(V-S)$ is the largest connected component in $V-S$.  We note that this is a smoothened version of the measure of vertex attack tolerance previously introduced by the authors in \cite{CASResSF}.  It will also be convenient to refer to this unsmoothened version of vertex attack tolerance, which we denote as follows:
\begin{equation*}
\hat{\tau}(G) = \min_{S \subset V} \{\frac{|S|}{|V-S-C_{max}(V-S)|} \}
\end{equation*}
Combinatorial conductance or edge based conductance is defined as
\begin{eqnarray*}
\Phi(G) &= \min_{S \subset V, Vol(S) \le Vol(V)/2} \{ \frac{|Cut(S,V-S)|}{Vol|S|} \} \\
 &= \min_{S \subset V,  Vol(S) \le Vol(V)/2} \{ \frac{|Cut(S,V-S)|}{\delta_S|S|} \} 
\end{eqnarray*}
where $|Cut(S,V-S)|$ is the size of the cut separating $S$ from $V-S$, $Vol(S)$ is the sum of the degrees of vertices in $S$, and $\delta_S$ is the \emph{average} degree of vertices in $S$.  Vertex expansion is defined as
\begin{equation*}
\epsilon^V(G) = \min_{S \subset V} \{ n\frac{|N(S)|}{|S||V-S|} \} 
\end{equation*} 
where $N(S)$ denotes the outer boundary of $S$, namely $N(S) = \{ i \in V-S \mid \exists u \in S,_\ni \{u,v\} \in E \}$.  Integrity is defined as
\begin{equation*}
I(G) = \min_{S \subset V} \{ |S|  + C_{max}(V-S) \} 
\end{equation*}
Toughness is defined as 
\begin{equation*}
t(G) = \min_{S \subset V}\left\{\frac{|S|}{\omega(V-S)}\right\}
\end{equation*}
where $\omega (V-S)$ is the number of connected components in V-S.  Tenacity is defined as 
\begin{equation*}
T(G) = \min_{S \subset V}\left\{\frac{|S|+C_{max}(V-S)}{\omega (V-S)}\right\}
\end{equation*} \\
Scattering number is defined as
\begin{equation}
sn(G) = \max_{S \subset V}\{ \omega(V-S)+|S| \}
\end{equation}\\
However, as higher scattering numbers corresponds to worse resilience, it will be more convenient to discuss the smoothed inverted version of scattering, which we denote
\begin{eqnarray*}
h(G) &= \frac{1}{sn(G)+1}
&= \min_{S \subset V} \{ \frac{1}{\omega(V-S)+|S|+1} \}
\end{eqnarray*}\\
\\
It is clear that all above notions are defined based on finding some \emph{worst case} target set of nodes.  Therefore, it will be convenient to directly denote the referenced target set as well as the measure defined conditional upon a particular target set.  For the case of vertex attack tolerance, we denote set-vertex tolerance as
\begin{equation*}
\tau_S(G) =\frac{|S|}{|V-S-C_{max}(V-S)|+1 }
\end{equation*}
so that clearly $\tau(G) = \min_{S \subset V} \tau_S(G)$ and correspondingly
\begin{equation*}
S(\tau(G)) = argmin_{S \subset V} \tau_S(G)
\end{equation*}
Similary for set-conductance:
\begin{equation*}
\Phi_S(G) = \frac{|Cut(S,V-S)|}{\delta_S|S|}
\end{equation*}
so that clearly $\Phi(G) = \min_{S \subset V, |Vol(S)| \le |Vol(V)|/2} \Phi_S(G)$ and correspondingly
\begin{equation*}
S(\Phi(G)) = argmin_{S \subset V, |Vol(S)| \le |Vol(V)|/2} \tau_S(G)
\end{equation*}
For set vertex expansion:
\begin{equation*}
\epsilon_S^V(G) = n\frac{|N(S)|}{|S||V-S|}
\end{equation*}
so that $\epsilon^V(G) = \min_{S \subset V} \epsilon^V_S(G)$ and
\begin{equation*}
S(\epsilon^V(G)) = argmin_{S \subset V} \epsilon^V_S(G) 
\end{equation*} 
The pattern is clear for set integrity, set toughness, set tenacity, and smoothed inverse scattering number, respectively denoted $I_S(G), t_S(G), T_S(G),$ and $h_S(G)$.  For all such set resilience measures $f_S$, we may generally denote the target set function as $S(f(G)) = argmin f_S(G)$.

It will be convenient to denote the subgraph of a graph $G = (V,E)$ that is induced by a vertex set $S \subset V$ as $G_S = (S, E_S)$.

Finally, there is a special infinite family of graphs which we shall refer to in comparing resilience notions, namely the \emph{star graphs} denoted $Star(n)$ and defined as follows:
\begin{defi}
The \emph{star graph} $Star(n)$ is an undirected graph $G = (V, E)$ on $n = |V|$ nodes with a unique designated central node $q$ such that node $q$ is connected to every other node $u \in V - \{ q \} $, and there are no other edges between any other pair of nodes.  Namely, $E = \{ \{q, u \} \mid u \in V - \{ q \} \}$.  When the labeling of the graph is arbitrary, we may without loss of generality assume that $q = 1$.
\end{defi}

\section{Theoretical Results}

As stated, conductance is arguably the most important and well-studied resilience notion in the context of $d$-regular graphs.  A main result of this section is the following theorem which states essentially, for any $d$-regular graph that does not exhibit too-high conductance, the vertex attack tolerance of the same graph cannot exceed the conductance by more than a factor of $d$:
\begin{thm}\label{thm:vatcond}
Given a connected, undirected $d$-regular graph $G = (V,E)$, let set $S = S(\Phi(G))$ such that the induced subgraph $G_S$ is connected. If $\Phi(G) \le \frac{1}{d^2}$ then $\tau(G) < d\Phi(G)$.
\end{thm}
In proving Theorem \ref{thm:vatcond} we shall use the following lemma regarding conductance in $d$-regular graphs:
\begin{lem}\label{lem:sconn}
Given a connected, undirected $d$-regular graph $G = (V,E)$, there exists a set $S$ such that $S = S(\Phi(G))$ and the induced subgraph $G_S$ is connected.
\end{lem}

Before presenting the proof of Theorem \ref{thm:vatcond} relating VAT via conductance in the case of $d$-regular graphs, we emphasize that conductance does not capture node-based resilience in the case of \emph{heterogeneous degree} graphs to any meaningful approximation factor.  We have motivated this comparison and contrast in previous work \cite{CASResSF} via example graphs, including the graph $Star(10)$, and in fact this discrepancy is a major motivation of this work itself.  However, here we wish to generalize and formalize this point theoretically, as we do in the next theorem:
\begin{thm}\label{thm:star}
There exists an infinite graph family, namely the star graphs $Star(n)$ such that a targeted attack against the central node results in a maximal disconnection of the graph into $n-1$ isolated nodes.  Yet, the following are the conductance and vertex attack tolerance of $Star(n)$ respectively, for $n \ge 3$:
\begin{enumerate}
\item $\Phi(Star(n)) = 1$
\item $\tau(Star(n)) = \frac{1}{n-1}$
\end{enumerate}
In other words, there is an infinite graph family which is maximally intolerant against a targeted node attack.  The conductance of this family is maximal, whereas the vertex attack tolerance of this family is minimal, amongst all possible graphs.
\end{thm}

In proving the above theorem, we will also make use of the following Lemma that characterizes minimal vertex attack tolerance of any graph:
\begin{lem}\label{lem:minvat}
For any connected, undirected graph $G = (V,E)$ on $n = |V| \ge 3$ nodes, $ \tau(G) \ge \frac{1}{n-1}$
\end{lem}

Now we demonstrate the proofs of Theorems \ref{thm:star} and \ref{thm:vatcond}, and Lemmas \ref{lem:minvat} and \ref{lem:sconn}, respectively:

\begin{proof}[Proof of Theorem \ref{thm:star}]
Let $q$ be the designated central node in $Star(n)$.  First, let us bound $\Phi(Star(n))$: Let $S = S(\Phi(Star(n)))$.  First consider the case that $q \ni S$.  In that case, all degrees of nodes in $S$ are one, and each node in $S$ is adjacent only to $q$ which is not in $S$.  Therefore, $|Cut(S,V-S)| = |S| = |Vol(S)|$, proving that $\Phi(Star(n)) = 1$ for this case.  In the other case that $q \ni S$, note that the volume constraint in the definition of conductance restricts us to sets $S$ that do not exceed half of the total volume of $V$.  The volume of $q$ alone is proportional to its degree $n-1$, whereas there are exactly $n-1$ other nodes whose total volume is exactly $n-1$.  Therefore, in the case that $q \in S$, it must be that $S = \{ q \}$ alone.  In this situation, $Cut(S, V-S)$ is exactly all the $n-1$ edges of $G$, and so calculating for conductance again yields $\Phi_S(Star(n)) = \frac{n-1}{n-1} = 1$.

Now we consider $\tau(Star(n))$.  As intuitively clear, first consider $\tau_S(Star(n))$ for $S = \{ q \} $, namely a targeted attack of the central node.  Because removal of $q$ results in only isolated nodes, $|C_{max}| = 1$.  Moreover, $|S| = 1$.  Therefore, $\tau_S(Star(n)) = \frac{1}{n-1}$.  By Lemma \ref{lem:minvat}, no other set can achieve a lower VAT value, and so $\tau(Star(n)) = \frac{1}{n-1}$.
\end{proof}

\begin{proof}[Proof of Theorem \ref{thm:vatcond}]
Let $S^{out}$ denote the vertex boundary of $S$ that is outside of $S$, also called the \emph{outer vertex boundary of $S$} and precisely being $S^{out} = \{ v \in V-S \mid \exists e = \{u,v\} \in Cut(S,V-S) \}$.  We may lower bound and upper bound $S^{out}$ as follows:
\begin{enumerate}
\item $|Cut(S,V-S)| \le d|S^{out}|$ with equality occurring only when all neighbors of nodes in $S^{out}$ lie in $S$ instead of $V-S$
\item $|S^{out}| \le |Cut(S,V-S)|$ with equality occurring only when all outer neighbors of nodes in $S$ are distinct
\end{enumerate}
Combining Bound (2) above with the fact that $\frac{|Cut(S,V-S)|}{d|S|} = \Phi(G) \le \frac{1}{d^2}$ we have that
\begin{equation*}
|S^{out}| \le |Cut(S,V-S)| \le \frac{|S|}{d}
\end{equation*}
Now, consider the structure of $G$ upon the removal of nodes $S^{out}$.  As $S^{out} \subset V-S$, and removal of $S^{out}$ also removes edges along $Cut(S,V-S)$, we know that at least $S$ would remain as a connected component.   There are two possible situations: either (i) $S$ would be the largest connected component in which case $|V-S^{out}-C_{max}(V-S^{out})| = |V-S-S^{out}|$ or, (ii) $S$ is not the largest connected component in which case $|V-S^{out}-C_{max}(V-S^{out})| \ge |S|$.  

First consider case (i): By definition of $\Phi$, $S$ cannot contain a strict majority of nodes of $V$, and therefore, $|V-S| \ge |S|$ is known.  Combining this with $|S^{out}| \le |S|\frac{1}{d}$ implies that
\begin{equation*}
|V-S^{out}-C_{max}(V-S^{out})| = |V-S-S^{out}| \ge |S|\frac{d-1}{d}
\end{equation*}
Further combining this with Bound (2) above, we have that
\begin{eqnarray*}
\tau_{S^{out}} &=  \frac{|S^{out}|}{|V-S^{out}-C_{max}(V-S^{out})|+1}  \\
&\le \frac{(d-1)|Cut(S,V-S)|}{d|S|+1} \\
&< (d-1)\Phi(G) < d\Phi(G)
\end{eqnarray*}
finishing the proof for case (i), because $\tau(G) \le \tau_{S^{out}} < d\Phi(G)$.

In the second case (ii), we also obtain that $\tau_{S^{out}} \le d\Phi(G)$: Due to the previously established facts $|S^{out}| \le |Cut(S,V-S)|$ and $|V-S^{out}-C_{max}(V-S^{out})| \ge |S|$ we have
\begin{eqnarray*}
\tau_{S^out} &= \frac{|S^{out}|}{|V-S^{out}-C_{max}(V-S^{out})|+1} \\
&< \frac{|Cut(S,V-S)|}{|S|} \\
&= d\Phi(G)
\end{eqnarray*}
finishing the proof.
\end{proof}

\begin{proof}[Proof of Lemma \ref{lem:minvat}]
Because $G$ is connected, at least one node must be attacked in order to result in any disconnection, so that the minimum achievable value of the numerator of the $\tau$ function is $1$.  Considering the maximum possible denominator of $\tau$, namely the maximum achievable $|V-S-C_{max}|+1$, we claim that it is $n-1$.  Because we already argued that $|S| \ge 1$, it suffices to show that $C_{max} \neq 0$ for this situation, as $n-1-1+1 = n-1$ then yields the desired lower bound.  But, $C_{max}$ cannot be zero unless $V = S$ (or else some node would remain), and taking $V=S$ gives a VAT of $1 >> \frac{1}{n-1}$ for $n \ge 3$.
\end{proof}

\begin{proof}[Proof of Lemma \ref{lem:sconn}]
Let $S = S(\Phi(G))$, and assume that $G_S$ can be partitioned into two subgraphs that are disconnected from each other, namely $G_1 = (S_1, E_1)$ and $G_2 = (S_2, E_2)$ such that $S_2 = S - S_1$ and $E_2 = E_S - E_1$.  Therefore, the edge boundary or $S$ is correspondingly partitioned such that 
\begin{equation*}
|Cut(S,V-S)| = |Cut(S_1,V-S)|+|Cut(S_2,V-S)|
\end{equation*}
For convenience, denote $a = |Cut(S_1, V-S)| = |Cut(S_1, V-S_1)|$, $b = |Cut(S_2, V-S)| = |Cut(S_2, V-S_2)|$, and $c = |Cut(S,V-S)| = a+b$, .  Furthermore, denote $x = |S_1|$, $y = |S_2|$, and $z = |S| = x+y$.  Now, consider the relative set conductances of $S_1$ and $S_2$, and without loss of generality, assume that $\Phi_{S_1}(G) \le \Phi_{S_2}(G)$.  Then, because the normalization factor $d$ is lost from both sides we have:
\begin{equation*}
\frac{a}{x} \le \frac{b}{y}
\end{equation*}
But, due to the minimality of the conductance achieved by $S$, we have
\begin{equation*}
\frac{a+b}{x+y} \le \frac{a}{x} \le \frac{b}{y}
\end{equation*}
But, it is verifiable that for any positive integers $a,b,x,y$ that if $\frac{a}{x} < \frac{b}{y}$ strictly, then 
\begin{equation*}
\frac{a}{x} < \frac{a+b}{x+y} < \frac{b}{y}
\end{equation*}
which would contradict the minimality of $\Phi(G) = \Phi_S(G)$.  Therefore, the previous inequality can only be satisfied when
\begin{equation*}
\frac{a+b}{x+y} = \frac{a}{x} = \frac{b}{y}
\end{equation*}
And, in that case, $S_1$ or $S_2$ also achieve the conductance of $G$.  However, we did not assume that either of them are connected.  Nonetheless, it is clear that as long as one of them is disconnected, the same argument as above can be applied to get another partition into two smaller sets yet which each achieve the conductance $\Phi$.  However, due to the finiteness of $G$, this cannot continue ad infinitum, and there must exist a set $S_i$ such that $\Phi_G = \Phi_{S_i}(G)$ and induced subgraph $G_{S_i}$ is connected.
\end{proof}

\section{Experimental Comparison of Resilience Measures}

Here, we experimentally compare and discuss the resilience measures of vertex attack tolerance, integrity\cite{Barefoot1987}, toughness\cite{Chvatal2006}, tenacity\cite{Cozzens1995} scattering number\cite{Jung1978}, and vertex expansion\cite{vertExpApprox,vertexSeparators}.  We have already established in the previous section that conductance is not an appropriate measure of resilience in the context of targeted node attacks.  The specific graphs we are comparing are shown in Fig.~\ref{fig:imageComparison}.  They are (a) $Star(10)$, (b) HOTnet \cite{HOTNet}, (c) C3, a graph that breaks logically into 3 similar sized components \cite{Costa}, (d) barbell, a 3-regular graph on 10 nodes with an obvious bottleneck, (e) a PLOD graph on 25 nodes \cite{Palmer2000}, (f) a 3-regular wheel on 10 nodes, and (g) a big barbell consisting of two cliques connected by a single edge.  The exact experimental measurements of the resilience notions for these graphs are presented in Table \ref{tab:comparison}.  The HOTnet and PLOD graphs are 25 nodes each.  The star, barbell, and wheel graphs are 10 nodes each.  Furthermore, the barbell and wheel graphs are also identical in the sizes of their edge sets, namely both have 15 edges.  The C3 graph has 33 nodes, whereas the big barbell graph has 24 nodes.


\begin{figure}
    \centering
    \subfigure[star]
    {
        \includegraphics[width=1.0in]{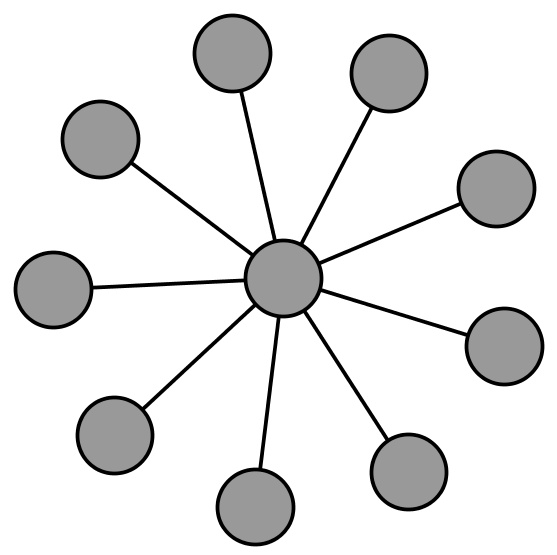}
        \label{fig:first_sub}
    }
    \subfigure[HOTNet]
    {
        \includegraphics[width=1.0in]{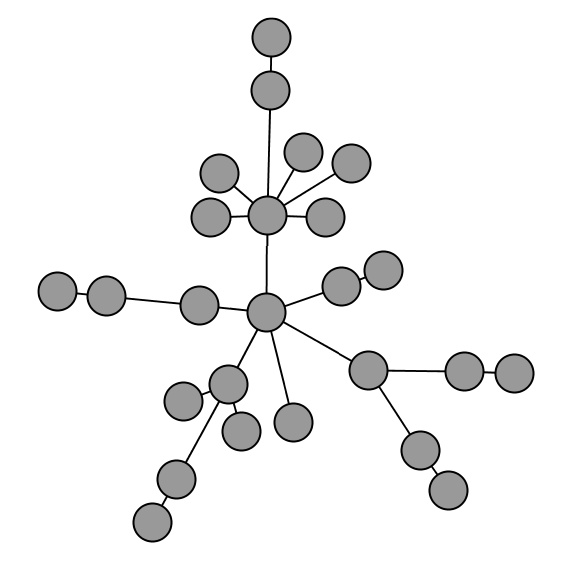}
        \label{fig:second_sub}
    } \\
    \subfigure[C3]
    {
        \includegraphics[width=1.0in]{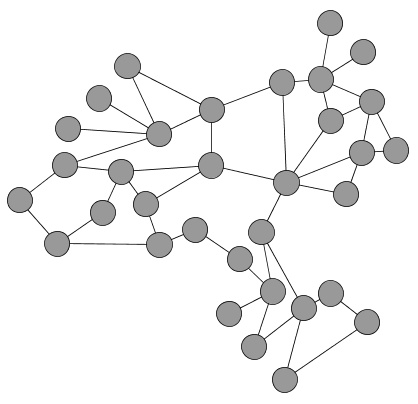}
        \label{fig:third_sub}
    }
\subfigure[barbell]
    {
        \includegraphics[width=1.0in]{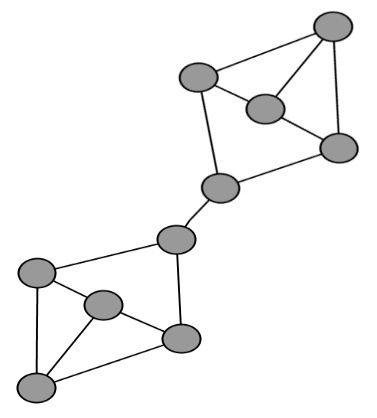}
        \label{fig:third_sub}
    } \\
 \subfigure[PLOD]
    {
        \includegraphics[width=1.0in]{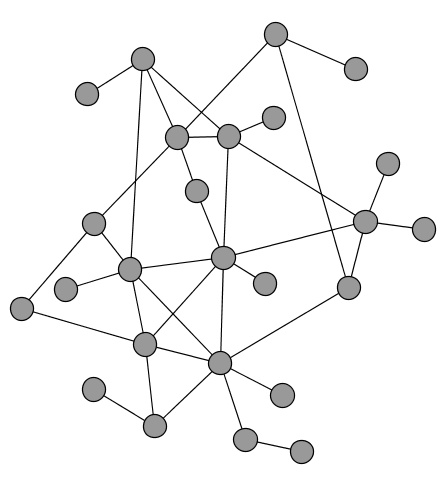}
        \label{fig:third_sub}
    }
\subfigure[wheel]
    {
        \includegraphics[width=1.0in]{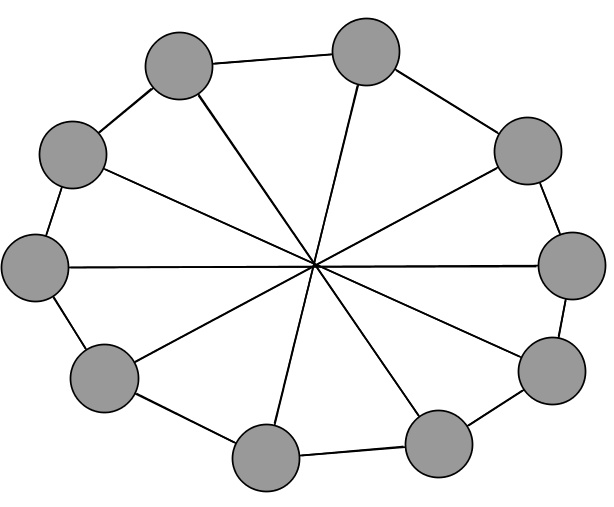}
        \label{fig:third_sub}
    } \\
\subfigure[big barbell]
    {
        \includegraphics[width=1.5in]{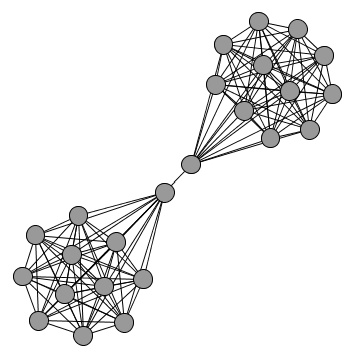}
        \label{fig:third_sub}
    }
    \caption{The graphs used in comparison.}
    \label{fig:imageComparison}
\end{figure}

\begin {table}
\caption {Comparison of measures on 7 graphs} \label{tab:comparison} 
\begin{center}
    \begin{tabular}{|l|r|r|r|r|r|c|}
    \hline
    Graph Type & VAT &  I(G) & t(G) & T(G) & h(g) &  $\epsilon^V(G)$ \\ \hline
    star           & .11 & .20  &  .11  & .22   & .11 & .40 \\
    barbell      & .20  & .60  &  .50 & 1.75    & .50 & .40 \\ 
    Wheel     & 1.00  & .60   &  1.00 & 1.2  & 1.00 & 1.60 \\ \hline
    HOTNet     & .06 & .28  &  .14  & .47   & .09 & .16 \\ 
    C3            & .15  & .36 &  .33 & .81    & .14 &  .29 \\ 
    big barbell     & .08  & .54  &  .50 & 6.5   & .50 & .17 \\ 
    PLOD       & .25  & .36  &  .33 & .80   & .17 &  .38 \\ 
    \hline
\end{tabular}
\end{center}
\end{table}

It is certainly meaningful to compare graphs of similar sizes to each other.  As such, we will first consider the star, barbell, and wheel graphs.  We claim that a measure that accurately captures resilience against targeted node attacks must clearly distinguish the resiliences of these three graphs, and rank their resilience in order of least resilient to most resilient as follows: star, barbell, and wheel.  We justify this as follows: The star graph is the most vulnerable graph because removal of a single node, which is a minimal size attack, creates a maximally disconnected remaining graph of only isolated nodes.  In fact, star graphs of any size are the most vulnerable graph of that size for the same reason.  With regards to the barbell graph, a single node attack (of either node adjacent to the single edge bottleneck) disconnects the graph into exactly half, which is a significant disconnection.  In the context of edge attacks, such a disconnection would be most significant because each edge removal can disconnect at most one new component.  However, in the context of vertex attacks, it is not as severe as a situation in which many small components results.  Nonetheless, amongst the set of attacks that results in exactly two remaining components, the worst case situation is for the components to be of similar sizes to each other (a half-way split), as that situation maximizes the disconnected pairs.  And, such a case is exhibited by the barbell graph.  As for the wheel, it is a very well connected graph for its size.  In fact, any removal of up to three nodes is only useful for disconnecting a single node from the remaining graph (a trivial disconnection).  In order to cause any non-trivial disconnection in the wheel, at least four nodes must be attacked, and four is already almost half of the size of the wheel itself.  If attacking half of the nodes of a graph is necessary to cause a serious disruption, then that graph is very resilient.  Therefore, we have justified our resilience ranking of these three graphs.

Now note that VAT clearly respects this resilience ranking of star, barbell, and wheel.  In fact, it can be verified by the table that \emph{all} the measures considered respect this ranking, though not strictly, except for tenacity, which ranks barbell as more resilient than wheel.  The other two problematic measures in the context of this ranking are integrity, which ranks barbell as equally resilient as wheel, and vertex expansion, which ranks star as equally resilient as barbell.  The problem with integrity is that the number of vertex removals necessary to cause a great disruption are not taken into account in any proportional manner.  The problem with vertex expansion on the other hand is that the distribution of \emph{component sizes} remaining is not taken into account, so that the worst case disruption always appears to correspond to a half-way split.  However, toughness and scattering number remain in tact as those measures respecting the rankings proposed thus far.

Now let us more closely examine toughness $t(G)$, scattering number $h(G)$, and VAT for the other graphs measured: All three measures agree upon HOTnet as being the least resilient of the remaining graphs.  This is reasonable by viewing the hierarchically star-like structure of the HOTnet illustrated, in addition to the fact that it is a tree (and so only $1$-connected graph).  VAT, toughness, and scattering also all agree upon graph C3 as being the next least-resilient graph.  However, comparing the resilience of the PLOD graph with both barbell and big barbell poses inconsistencies amongst the three measures: Both toughness and scattering rank both big barbell and barbell as having significantly higher resilience than PLOD (with both ranking barbell equally resilient as big barbell), but VAT ranks PLOD as being more resilient than barbell, which it ranks as more resilient than big barbell.  Both barbell and big barbell are graphs of severe bottlenecks in which the removal of one or two nodes creates a half-way split in the entire graph.  It is clear that PLOD has no such bottlenecks in the sense that there is no small subset of nodes whose removal causes a half-way or even three-way split.  However, PLOD does appear to be composed of a well connected core subgraph of the vast majority of nodes along with several tiny external subgraphs attached around this core, which we may refer to as the external \emph{hairs} connected to \emph{follicles} of the core graph.  Although there exists sets of follicle nodes whose removal may disconnect proportionally as many hair subgraphs, each new component that results is tiny in comparison to the core subgraph, and indeed indeed proportionally many follicles need to be attacked in order to disconnect them.  Therefore, we claim that PLOD should be considered more resilient than the barbell and big barbell graphs as consistent with VAT, as opposed to toughness and scattering.  Furthermore, whereas part of the problem in comparing these graphs is the size difference (which may explain the slight difference in the VAT measures as well), there is a more fundamental problem with toughness and scattering in their lack of capturing a possible bottleneck in a graph.  Note that this is the opposite problem of the vertex expansion (and conductance) measures, which failed to capture any information on remaining component sizes.  In the context of all the measures considered herein, and given all the graphs considered, we believe that VAT best captures graph resilience with respect to targeted node attacks.

\section{Comparison and Discussion of VAT of different Scale-Free models}

Because exact calculation of vertex attack tolerance for sufficiently large graphs proved to be intractable, we used a combination of techniques to obtain either exact values or tight upper bounds.  A branch-and-bound technique was used on small graphs to get exact values, and this proved infeasible for graphs larger than 35 nodes.  For larger graphs, we formulated a genetic algorithm to obtain close approximations of the VAT values.  While the margin of error of the genetic algorithm cannot be ascertained, the results are certainly valid upper bounds of the actual VAT values.  We will explain the genetic algorithm used in more detail in the subsequent section both because of some novelty in our method that might generalize to optimize genetic algorithms to compute other resilience measures, and also so that the reader can put our results in perspective.  The results obtained by running the genetic algorithm on the preferential attachment based Barabasi-Albert graphs (with $m = 2$ neighbors chosen for each entering vertex) and the random scale-free graphs of the same degree distribution are summarized in Table \ref{tab:vatcomp}.  Precisely, the \emph{degree distribution} of the random scale-free graphs were \emph{generated} by the corresponding preferential attachment graph\footnote{This is another reason why we omit the HOTnet, as we cannot guarantee generating a HOTnet with an exactly pre-specified degree distribution.}.  Because we are interested in how the topological generative properties affects the resilience of different scale-free models, we felt it is particularly important to control for degree distribution as much as possible.  Moreover, because a random graph of any given degree distribution can be defined, it is appropriate in this case to first generate the BA graph and then feed its degree distribution into the random scale-free graph generator.  We note further that our random scale-free graph generator is thus identical to the PLOD model \cite{Palmer2000} except for our explicit input of the degree distribution.

It can be seen that the random scale-free graphs have consistently better vertex attack tolerance than the Barabasi-Albert graphs for the parameter $m = 2$.  We also computed VAT values for various HOTnets, and all appeared significantly less resilient than the BA and random scale-free graphs of corresponding size.  We computed BA graphs with parameter $m = 1$ as well and discovered that such BA graphs were significantly non-resilient, with easily identifiable hubs whose removal causes severe disruption, yielding very low VAT values.  We had difficulty in generating identical degree PLOD graphs from BA graphs with $m =1$ due to the preponderance of degree one nodes.  Additionally the generated PLOD graphs tended not to be connected, so that the VAT value was automatically meaningless.  Therefore, we include the results for the $m = 2$ BA graphs and their corresponding degree PLOD counterparts only, noting that such graphs remain relatively resilient despite doubling of graph sizes.  However, the corresponding PLOD graphs remain \emph{more} resilient than the BA graphs of identical degree distribution, provided they are connected (which happens in likelihood given $m = 2$).

In addition to the actual computer VAT values in the table \ref{tab:vatcomp}, one may view the worst-case disruptions caused by the critical attack sets of the BA graph and the identically sized and distributed PLOD graph of 1000 nodes respectively in Figures \ref{fig:baattackviz} and \ref{fig:plodattackviz}.  It can be seen that the worst case proportional disruption to the PLOD graph only appears to remove a hair off of a well-connected core whereas the worst-case proportional disruption to the BA graph disconnects many components, some of which are sizeable, at once.

\begin {table}
\caption {Comparison of VAT of BA and Random SF graphs} \label{tab:vatcomp} 
\begin{center}
    \begin{tabular}{|l|r|r|r|r|r|c|}
    \hline
    n     & B-A  &  Random \\ \hline
    40    & 0.3000    & 0.33333  \\
    45    & 0.2692    & 0.30435  \\
    100   & 0.216667  & 0.26087  \\
    250   & 0.186813  & 0.25000  \\ 
    500   & 0.19171   & 0.28571  \\ 
    1000  & 0.193289  & 0.22222  \\ 
    2500  & 0.192679  & 0.22222  \\ 
    \hline
\end{tabular}
\end{center}
\end{table}

\begin{figure}
    \centering
    \subfigure[BA Attack]
    {
        \includegraphics[width=3.0in]{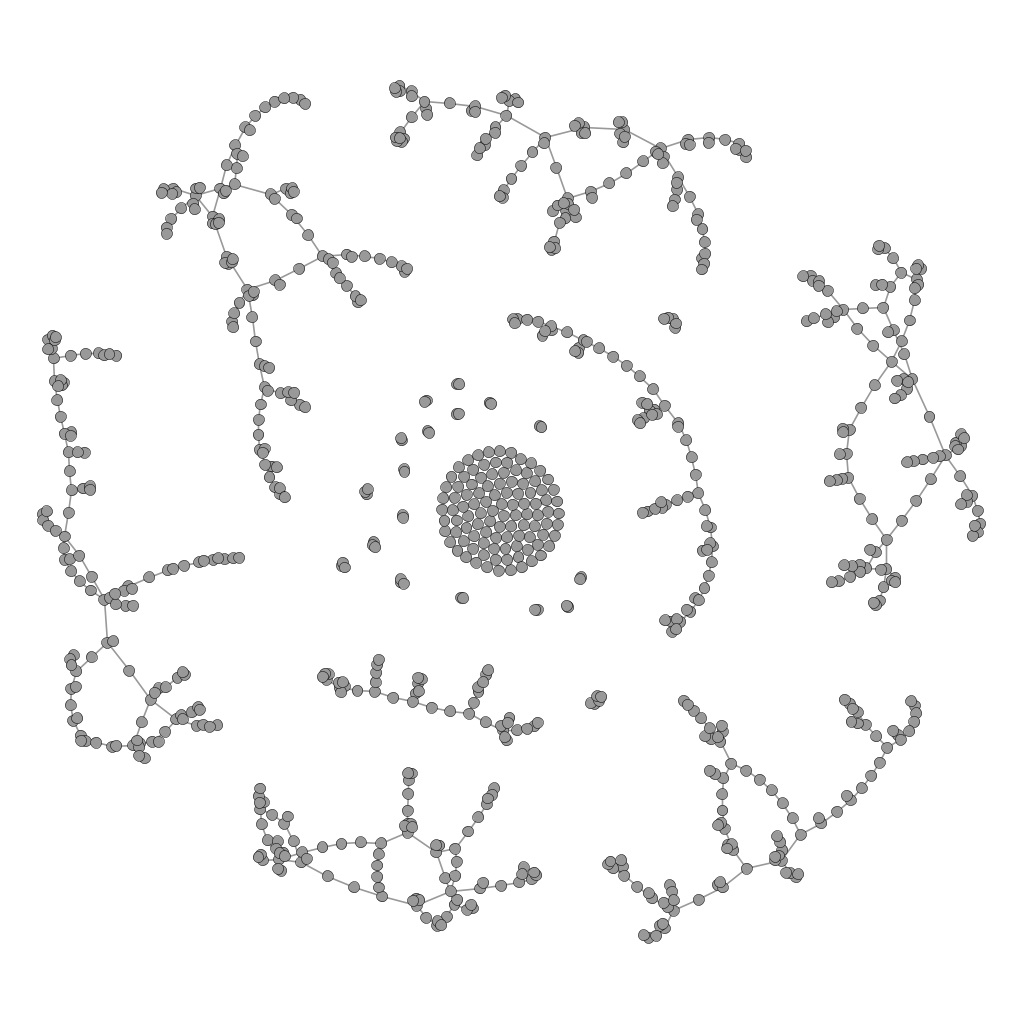}
        \label{fig:baattack}
    }
    \subfigure[BA CC Sizes]
    {
        \includegraphics[width=3.0in]{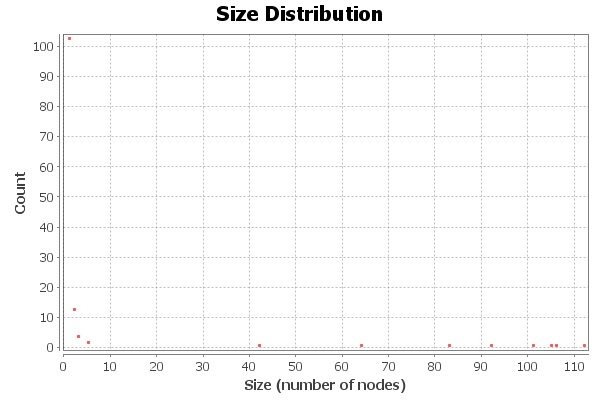}
        \label{fig:baccsizes}
    }
    \caption{BA attack visualization and resulting component size distributions.}
    \label{fig:baattackviz}
\end{figure}

\begin{figure}
    \centering
    \subfigure[PLOD Attack]
    {
        \includegraphics[width=3.0in]{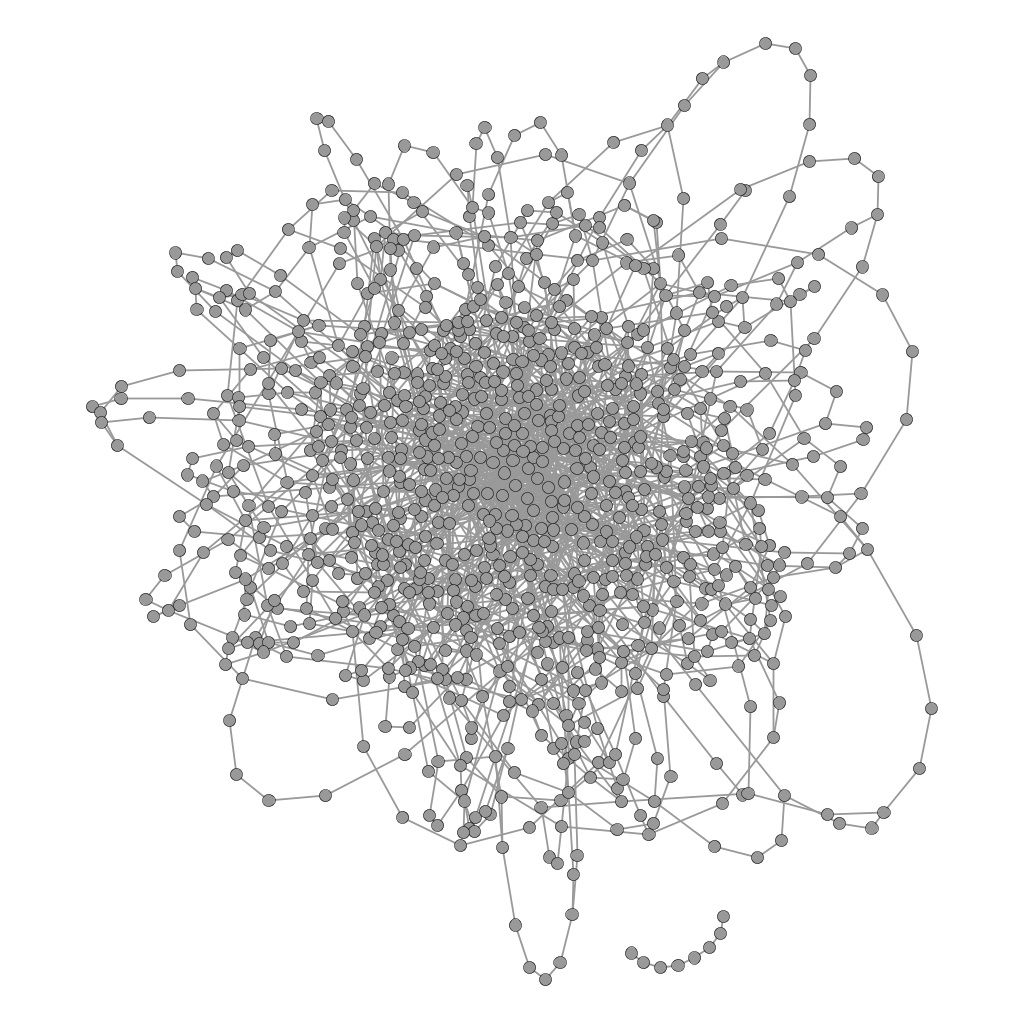}
        \label{fig:plodattack}
    }
    \subfigure[PLOD CC Sizes]
    {
        \includegraphics[width=3.0in]{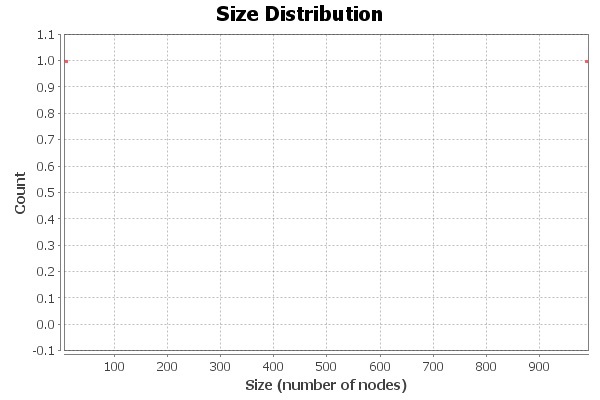}
        \label{fig:plodccsizes}
    }
    \caption{PLOD attack visualization and resulting component size distributions.}
    \label{fig:plodattackviz}
\end{figure}

\section{Computation of VAT}

As with many other resilience notions, calculating a global minimum for VAT proves to be a challenging endeavor.  As $V$ grows larger, the size of the search space for a set S, $|\mathcal P \left({V}\right)|$, minimizing vertex conductance grows exponentially.  Given this, exact computation over all $S$ proves intractable for $V > 35$.  To overcome this, a genetic programming approach was explored.

\subsection{Genetic Algorithm}
The genetic algorithm consists of the following steps: First, a set of random chromosomes is generated. Each chromosome $c_j$ consists of a set of $|V|$ bits $B$.  Let $V = \{1, 2, \dots, n \}$ be the usual ordering of nodes.  The situation that $B_i$ is $1$ corresponds to considering $V_i \in S$, where $S$ is the candidate target set of vertex attack tolerance. Next, a fitness $f_j \in [0,1]$ is assigned to each chromosome, corresponding to $1 - \tau_{c_i}(G)$.  Note that because vertex attack tolerance is normalized, $f_j$ is also normalized.  After fitness computation, a new population of chromosomes is generated: Two mates are selected using tournament selection, favoring individuals with higher fitness.  A crossover operation is performed with some probability to generate 2 new chromosomes.  Each new chromosome undergoes a mutation step, wherein each bit is flipped with some non-zero probability.  The fitness is then evaluated for the new population and the cycle continues.

While this technique is more feasible for large graph sizes, two observations are noted.  First, the high rate of memory access bottlenecks GPU based parallelization.  Secondly, beneficial mutations are seemingly rare and produce a sub-optimal population.

\subsection{Hill-Climbing}
With these insights, a hill-climbing approached was tried to enhance the genetic algorithm.  One initial chromosome is generated using a randomized contraction based min-cut algorithm \cite{Karger1993}, guaranteeing a cut.  To transform this edge set into a node set $S$, one endpoint is selected at random from each edge and added to the node set.  This guarantees that $V-S$ will be disconnected.  A mutation vector $M$ of length $j$ is created with $m_1 < m_2 < \dots < m_j \le |V|$ initialized to $<1,2, \dots,j>$.  A mutated chromosome with Hamming distance $j$ is generated by flipping bits $m_1, m_2, \dots, m_j$.  The fitness is then evaluated. If the new chromosome is more fit, then $M$ is reset and the new chromosome is kept; otherwise, $m_j$ is incremented using ripple carry to deal with $m_j = |V|+1$. This technique attempts all possible mutations of length $j$ then increments $j$ when no progress is made.  To optimize progress, mutation length $j<3$ is used and many vertex cuts are tried.

While the genetic algorithm performed relatively well, beneficial mutations become increasingly improbable.  This newer non-probabalistic method is able to improve chromosomes signifigantly beyond the point they stagnate in the genetic algorithm.  After 200,000 generations with population size 64, BA with 100 verticies stagnated.  Hill-climbing decreased this an additional 10 percent in minutes.    Not only does it acheive better bounds on VAT, but it does so orders of magnitude faster than the parallelized GA.

\section{Discussion and Continuing Work}

We have motivated a measure of resilience for networks in the context of targeted node attacks, which we termed vertex attack tolerance or VAT, compared our measure against existing measures, and compared the VAT of preferential attachment and random scale-free graphs.  Our theoretical results indicate that VAT behaves similarly to conductance in the case of $d$-regular graphs, but captures node-based resilience in a way that conductance cannot in the case of heterogeneous degree graphs.  Our comparison of VAT with the other resilience measures indicates that VAT is the only measure that fully captures both the major \emph{bottlenecks} of a network and the resulting \emph{component size distribution} upon targeted node attacks.  

With regards to the comparative node-based resilience of different scale-free models, the preferential attachment model (B-A model) and the random scale-free model in particular, our results indicate that the random scale-free networks are consistently more resilient against targeted node attacks than the preferential attachment networks of the exact same degree distribution.  However, somewhat surprisingly, both network types appear relatively resilient compared to their sizes: Even as the network continues to double the VAT values decrease very slowly.  We also noted that the HOTnets appear significantly less resilient than both BA and random scale-free graphs.  In evaluation of our results, there is the caveat that results were obtained via a mixed genetic and hill-climbing algorithmic approach so that we cannot guarantee the accuracy of the measures.  However, we allowed the algorithm to run for a minimum of 100,000 generations of size 64 taking 1-8 hours, particularly for large graphs.  Furthermore, hill-climbing up to a mutation length of 2 was performed on both the previous GA results and a minimum of 1,000 random vertex cuts.  Therefore, the accuracy of the results have been enhanced as much as possible.

Nonetheless, in the course of our computations, we discovered that with parameter $m = 2$ finding a \emph{bad} vertex set $S \subset V$ whose removal causes notably greater disruption in the network compared to other vertex sets became significantly more difficult, especially for the BA type graphs.  In terms of VAT, we are referring to sets $S$ whose behavior is similar to the set $S(\tau(G))$.  Much previous work discusses ``hubs'' of such networks, easily identifiable by the highest degree nodes, which are purported major points of vulnerability.  We found that while such hubs are absolutely critical for BA graphs of parameter $m = 1$, and they are also important for BA graphs with $m = 2$, looking only at the hubs is insufficient to find the critical attack sets for BA graphs of $m = 2$.  On the other hand, once the critical attack set is discovered for any of the BA type graphs, indeed the level of disruption caused is more significant compared to the worst case disruption of the PLOD graph of corresponding degree.  Nonetheless, finding the non-hub nodes in the most critical attack set of the BA graphs of $m = 2$ was very difficult and time consuming, and most attack sets caused little to no harm for such graphs.  While this further validates the tolerance of such graphs against \emph{random} vertex attacks, it also invokes the following question: Given the difficulty of discovering the worst case set of nodes with respect to VAT, under what reasonable computational and information-theoretic assumptions can an \emph{attacker} discover such a set of nodes for a given scale-free graph model?

We continue to work on defining, bounding, and computing vertex attack tolerance \emph{with respect to an attack model $F$} for various network types.  For example, the simplest case of a \emph{random} attack model $R$ yields $\tau^R(G) = E[\tau_S(G)]$, the expected value of VAT over all subsets of vertices.  Other important attack models naturally include those preferring the highest degree nodes as well as those preferring the highest betweenness nodes.  One problem with a model based on betweenness and other centrality notions, however, is that they are not in general a \emph{locally computable} quantity, but rather requires global information about the network.  In future work, we consider the power of an attacker or set of attackers with respect to VAT, when their informational and computational resources are reasonably limited.

\bibliographystyle{plain}
\bibliography{res-sf}

\end{document}